\documentclass[preprint,12pt]{elsarticle}
\makeatletter
\def\ps@pprintTitle{%
 \let\@oddhead\@empty
 \let\@evenhead\@empty
 \def\@oddfoot{\centerline{\thepage}}%
 \let\@evenfoot\@oddfoot}
\makeatother
\usepackage{soul}
\usepackage{graphics}
\usepackage{graphicx}
\usepackage{psfig}
\usepackage{rotating}
\usepackage[absolute]{textpos}
\usepackage{color}
\usepackage{lineno}
\usepackage[utf8]{inputenc}
\usepackage{amsmath}
\usepackage{amsfonts}
\usepackage{amssymb}
\usepackage{amsthm}
\usepackage{mathrsfs}
\usepackage{latexsym}
\usepackage{multirow}
\usepackage{rotating}
\usepackage{caption}
\usepackage{graphicx}
\usepackage{float}
\usepackage{array}
\usepackage{subfigure}
\usepackage{tabularx}
\usepackage{graphicx,psfig}
\usepackage{footnote}
\usepackage[titletoc,toc,title]{appendix}
\usepackage{colortbl}
\usepackage[dvipsnames]{xcolor}
\usepackage{hyperref}
\usepackage{natbib}
\usepackage{amssymb}
\usepackage{soul}
\usepackage{algorithmic}
\usepackage{algorithm}
\usepackage{epsfig}
\usepackage{relsize}
\usepackage{color}
\usepackage{mathtools}
\usepackage{multirow}
\usepackage[utf8]{inputenc}
\usepackage[english]{babel}

\newtheorem{theorem}{Theorem}

\newtheorem{prop}{Proposition}
\newtheorem*{remark}{Remark}

\begin{document}

\title{A novel distribution-free hybrid regression model for manufacturing process efficiency
improvement}
\author{Tanujit Chakraborty\footnote[1]{\textit{Corresponding author}:
Tanujit Chakraborty (tanujit\_r@isical.ac.in)}, Ashis Kumar Chakraborty\textsuperscript{2}, Swarup Chattopadhyay\textsuperscript{3}\\
    {\scriptsize \textsuperscript{1, 2 and 3} Indian Statistical Institute, 203, B. T. Road, Kolkata - 700108, India}}
    %{\scriptsize \textsuperscript{3}Machine Intelligence Unit, Indian Statistical Institute, 203, B. T. Road, Kolkata - 700108, India}}
%\author{Tanujit Chakraborty, Ashis Kumar Chakraborty, Swarup Chattopadhyay}
%\address{SQC and OR Unit, Indian Statistical Institute, Kolkata, India}
%\emailauthor{tanujit\_r@isical.ac.in}{Tanujit Chakraborty}
\begin{abstract}
This work is motivated by a particular problem of a modern paper
manufacturing industry, in which maximum efficiency of the
fiber-filler recovery process is desired. A lot of unwanted
materials along with valuable fibers and fillers come out as a
by-product of the paper manufacturing process and mostly goes as
waste. The job of an efficient Krofta supracell is to separate the
unwanted materials from the valuable ones so that fibers and fillers
can be collected from the waste materials and reused in the
manufacturing process. The efficiency of Krofta depends on several
crucial process parameters and monitoring them is a difficult
proposition. To solve this problem, we propose a novel hybridization
of regression trees (RT) and artificial neural networks (ANN),
hybrid RT-ANN model, to solve the problem of low recovery percentage
of the supracell. This model is used to achieve the goal of
improving supracell efficiency, viz., gain in percentage recovery.
In addition, theoretical results for the universal consistency of
the proposed model are given with the optimal value of a vital model
parameter. Experimental findings show that the proposed hybrid
RT-ANN model achieves higher accuracy in predicting Krofta recovery
percentage than other conventional regression models for solving the
Krofta efficiency problem. This work will help the paper
manufacturing company to become environmentally friendly with
minimal ecological damage and improved waste recovery.
\end{abstract}


\maketitle
%\end{frontmatter}

\section{Introduction} \label{Introduction}
Regression problems arise in many practical situations where a
specific response variable can be expressed through a relationship
with the so-called causal variables. In practical applications, it
becomes quite challenging to identify the right set of causal
variables. This article is motivated while we were dealing with a
problem in the modern paper industry. It is well known that paper is
produced from pulp, fibers, fillers, and other chemicals and it
needs a considerable water resource as well. As it is, a lot of
fibers and fillers obtained from the base material are drained out
along with outlet water that comes from the paper machine. To save
these valuable materials from being drained out, a system called
Krofta supracell \cite{krofta1990water} is used by modern paper
manufacturing industries and many other similar types of industries.
An efficient recovery system of fibers and fillers from the outlet
flow will naturally be cost saving and help the industry to remain
competitive in the market. This is also known as the fiber-filler
recovery process which is practically a dissolved air flotation cum
sedimentation process \cite{koukoulas2003method}. An input to this
process is the lean backwater which comes out of the main
papermaking process and hence not only contains a lot of unwanted
materials but also contains fibers and fillers which forms the basis
of making any paper. Collection of such fibers and fillers
efficiently will save cost and protects the environment where
finally these materials are thrown in as wastage. By addressing this
problem, one can make the paper manufacturing process an
environment-friendly and sustainable manufacturing process with
minimal ecological damage also. In manufacturing processes,
particularly processes of chemical types, knowledge of the causal
variables which affect the Krofta efficiency, is very limited and
one has to start from scratch to gather such a knowledge
\cite{chen1998treatment}.

A lot of preliminary discussions and analysis helped to determine a
set of possible variables, some of which are expected to be critical
causal variables for Krofta efficiency improvement. Details of these
variables are given in Section 4.2. The focus of the problem was to
judge and decide statistically which variables need to be controlled
to improve the Krofta efficiency to at least 80\% from the present
level of 60\%. While trying to develop a methodology to ensure the
targeted improvement, one needs to think about the robustness of the
method developed. To ensure robustness, we took recourse to
distribution-free regression models. Nonparametric statistical and
machine learning models like decision trees, random forest, support
vector regressions (SVR) and artificial neural networks, etc have
been applied to solve several problems in analyzing water quality of
river \citep{singh2009artificial, mahuli1993ph}, surface water
planning \citep{gmar2017electrodialytic, bhattacharya2005neural} and
urban water demand forecasting \citep{brentan2017hybrid,
sebri2016forecasting}. It was found that decision trees and neural
networks (NN) can model arbitrary decision boundaries. To utilize
the positiveness of these two models, theoretical frameworks for
combining both the models are often used jointly to make decisions.
The idea of mapping tree-based models into ANN are presented by
several papers to solve many supervised learning problems
\citep{sethi1990entropy, sirat1990neural}. Similar designs in the
interface area of decision trees and neural networks can be found in
some recent literature as well \citep{chen2005time,
balestriero2017neural}. In spite of the use of neural tree models in
practical problems of classification, regression, and forecasting,
very little is known about the statistical properties of these
models. Thus, the significant disadvantages of these algorithms are
having poor robustness and having many free tuning parameters.

Motivated by the above discussion, we propose in the present paper a
hybrid RT-ANN model and discuss its statistical properties with the
optimal values of the model parameters. In this model, a consistent
RT algorithm is used as an important causal variable selection
algorithm. Further, a neural network with the input variables chosen
from RT along with RT algorithm outputs is trained to minimize the
empirical risk on the training set which results in a universally
consistent hybrid model. We have then introduced the consistency
results of hybrid RT-ANN model which assures a basic theoretical
guarantee of robustness of the proposed algorithm. The proposed
hybrid RT-ANN model has the advantages of significant accuracy,
converges much faster than complex hybrid models and easy
interpretability as compared to more ``black-box-like" advanced
neural networks. Also, no parametric assumptions are made on the
distribution of the input and output variables in our proposed
model. Our model has the advantages of less number of tuning
parameters and is useful for high-dimensional small data sets and
complex data structures. Unlike other regression models, hybrid
RT-ANN model doesn't have any strong assumption about the normality
of the data and homoscedasticity of the noise terms. This model is
applied to solve the Krofta efficiency problem of the paper
manufacturing company to achieve a specific target for improving
supracell efficiency. Through regression modeling, we obtained the
necessary control parameters and this model also explains the Krofta
recovery percentage with higher accuracy than other conventional
regression models. Based on the recommendation of the regression
analysis, an experimental design is run to find optimum levels of
control parameters which are then implemented in the process and
results in better recovery than the targeted retention of valuable
materials. Our recommendations not only results in an economic gain
and wastage reduction but also improved production process and
environmental benefits.

The paper is organized as follows. In Section 2, we introduce hybrid
RT-ANN model and its statistical properties are discussed in Section
3. Usefulness of the proposed model in improving manufacturing
process efficiency is shown in Section 4. Finally, the concluding
remarks are given in Section 5.

\section{Proposed Hybrid model} \label{Proposed_hybrid_model}

One of the ultimate goals of designing a regression model is to
select the best possible regressors which can predict the response
variable accurately. Regression trees (RT) and artificial neural
networks (ANN) are competitive techniques for modeling regression
problems. RT is a hierarchical nonparametric model
\cite{breiman2017classification} relatively superior to ANN in the
readability of knowledge. But, ANNs are better in implementing
comprehensive inference over the inputs \cite{tokar1999rainfall}.
General sufficient condition for the consistency of data-driven
histogram-based regression estimates is when the size of the tree
grows with the number of samples at an appropriate
rate\cite{nobel1996histogram}. And it was theoretically proven that
if a one-layered neural network is trained with an appropriately
chosen number of neurons to minimize the empirical risk on the
training data, then it results in a universally consistent neural
network estimates \cite{lugosi1995nonparametric}.

In the proposed model, we have used RT as a feature selection
algorithm \cite{narendra1977branch} which has a built-in mechanism
to perform feature selection \cite{quinlan1993c4}. RT as a Feature
selection algorithm has many advantages such as its ability to
select important features from high dimensional feature sets
\cite{dash1997feature}. Feature selection using RT algorithm can be
characterized by the following:

\begin{itemize}
\item With the set of available features as an input to the algorithm; a tree is created.
\item RT has leaf nodes, which represent predicted regression output value.
\item The tree branches represent each possible value of the feature node from which it originates.
\item RT is used to choose feature vectors by starting from its root and moving through it until a leaf node is identified.
\item At each decision node in the RT, one can select the most useful feature using appropriate estimation criteria.
The criterion used to identify the best possible features involves
mean squared error for regression problems.
\end{itemize}

In our proposed methodology, we first split the feature space into
several areas by RT algorithm. Most important features are chosen
using RT and redundant features are eliminated. Further, we build
the ANN model using the important variables obtained through RT
algorithm along with prediction results made by RT method which are
used as another input information in the input layer of neural
networks. We then run the ANN model with one hidden layer and the
optimum value of the number of neurons in the hidden layer is
proposed in Section 3. Since we have taken RT output as an input
feature in ANN model, the number of hidden layers are chosen to be
one for further modeling with ANN algorithm, and it is also shown to
be universally consistent with some regularity conditions in Section
3. The effectiveness of the proposed hybrid RT-ANN model lies in the
selection of important features and using regression outputs of the
RT model followed by the ANN model. The inclusion of RT output as an
input feature increases the dimensionality of feature space and will
also increase class separability \cite{chakraborty2018novela}. The
informal work-flow of our proposed hybrid model, shown in Figure 1,
is as follows.
\begin{itemize}
\item  First, we apply a regression tree algorithm to train and build a decision tree model that calculates important features as a feature selection algorithm.
\item  The prediction result of RT algorithm is used as an additional feature in the input layer of the ANN model.
\item  We use important input variables obtained from RT along with an additional input variable (RT output) to develop an appropriate ANN model, and the network is generated.
\item  Run the ANN algorithm with sigmoid activation function and the optimal number of hidden layer neurons as shown in Section 3 and record the regression outputs.
\end{itemize}
This algorithm is a two-step problem-solving approach such as
feature selection using RT and then using all the outputs of feature
selection algorithm in the following regression analysis to get a
better model in terms of accuracy. Our proposed model can be used
for identifying important features which will satisfy a specific
goal to solve Krofta efficiency problem and also can be employed for
modeling Krofta recovery percentage in terms of important causal
process parameters. Though the model was developed primarily to
solve the Krofta efficiency problem, can be used in other similar
situations as well.

On the theoretical side, it is necessary to show the universal
consistency of the proposed model and other statistical properties
for its robustness. We will now introduce a set of regularity
conditions to prove the risk consistency of RT as feature selection
algorithm and the importance of RT output in a further model
building with the ANN algorithm in the proposed hybrid model.

\begin{figure}[t]
\centering
%\subfloat[]
%\subfig{
%\includegraphics[height=1cm,angle=0,width=.6\linewidth,]{mod_realworld.eps}%
  %\epsfig{file=amazon_sim_final.ps, height=2.5in, width=3in}
%\includegraphics[width=0.6\textwidth,natwidth=1250,natheight=1250]{mod_real_world.PNG}
\includegraphics[scale=0.50]{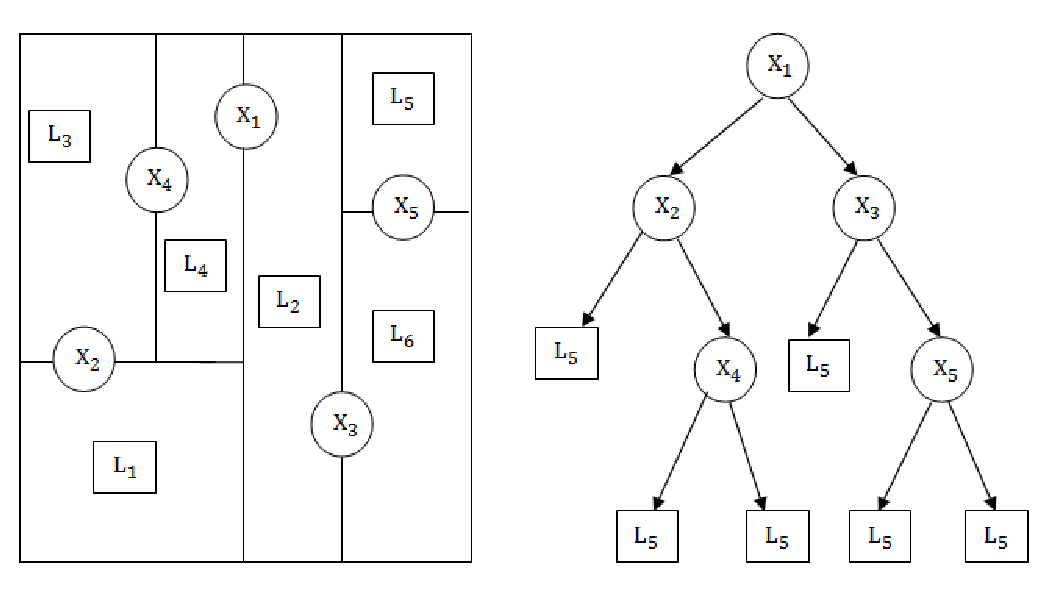}
\includegraphics[scale=0.50]{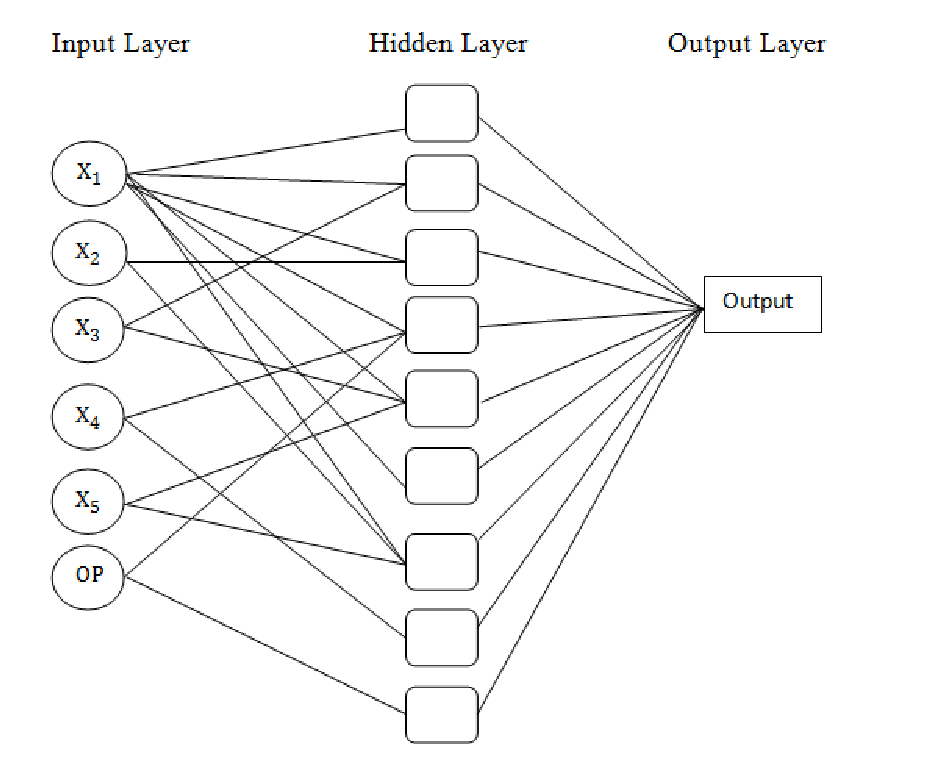}
\caption{An example of hybrid RT-ANN model with $X_{i}; i=1,2,3,4,5$
as Important features obtained by RT, $L_{i}$ as leaf nodes and OP
as RT output.} \label{figmodRWNGT}
\end{figure}

\section{Statistical Properties of Hybrid RT-ANN model} \label{Statistical Properties of Hybrid RT-ANN model}

RT has a built-in mechanism to perform feature selection
\cite{breiman2017classification,quinlan1993c4} and it has the
ability to select important features from high dimensional feature
sets \cite{narendra1977branch,dash1997feature}. To explore its
statistical properties, we are going to investigate the sufficient
condition(s) for consistency of RT. For a wide range of
data-dependent partitioning schemes, the consistency of
histogram-based regression estimates was shown in the literature
\cite{nobel1996histogram}. It requires a set of regularity
conditions to be satisfied to show the consistency of histogram
regression estimates. Also, it is assumed to have regression
variables to be bounded throughout. But in our case, we will
represent regression trees where the partitions are chosen to have
rectangular cells, i.e., regression trees employing axis-parallel
splits and response variable can take values within a certain range.
There is no assumption made on the distributions of predictor and
response variables.

Let $\underline{X}$ be the space of all possible values of $p$
features, i.e. $\underline{X}=(X_{1},X_{2},...,X_{p})$ and
$\underline{Y}$ be set of all possible values that the response
variable can take and $Y \in [-K,K]$. We wish to estimate regression
function $r(x)= E(Y | X = x) \in [-K,K]$ based on a training sample
$L_{n}=\{(X_{1},Y_{1}), (X_{2},Y_{2}),...,(X_{n},Y_{n})\}$, where
$X_{i}=(X_{i1},X_{i2},...,X_{ip}) \in \underline{X}$ with n observations and let
$\Omega=\{\omega_{1},\omega_{2},...,\omega_{k}\}$ be a partition of
the feature space $\underline{X}$. We denote $\widetilde{\Omega}$ as one such
partition of $\Omega$. Define
$(L_{n})_{\omega_{i}}=\{(X_{i},Y_{i})\in L_{n}: X_{i}\in \omega_{i},
Y_{i}\in [-K,K]\}$ to be the subset of $L_{n}$ induced by
$\omega_{i}$ and let $(L_{n})_{\widetilde{\Omega}}$ denote the
partition of $L_{n}$ induced by $\widetilde{\Omega}$.

The criterion used to identify the best features involves mean
squared error for regression problems. We don't make any assumption
on the distribution of the pair $(\underline{X},\underline{Y}) \in
\mathbb{R}^{p} \times [-K,K]$. Mean squared error is used to
partition feature space into a set $\widetilde{\Omega}$ of nodes.
There exists a partitioning regression function
$d:\widetilde{\Omega} \rightarrow Y$ such that $d$ is constant on
every node of $\widetilde{\Omega}$. Now let us define
$\widehat{L}_{n}$ be the space of all learning samples and
$\mathbb{D}$ be the space of all partitioning regression function
then, binary partitioning and regression tree based rule
$\Phi:\widehat{L}_{n} \rightarrow \mathbb{D}$ such that
$\Phi(L_{n})=(\psi \circ \phi)(L_{n})$, where $\phi$ maps $L_{n}$ to
some induced partition $(L_{n})_{\widetilde{\Omega}}$ and $\psi$ is
an assigning rule which maps $(L_{n})_{\widetilde{\Omega}}$ to a
partitioning regression function $d$ on the partition
$\widetilde{\Omega}$. The most basic reasonable assigning rule
$\psi$ is the plurality rule
$\psi_{pl}((L_{n})_{\widetilde{\Omega}})=d$ such that if $x \in
\omega_{i}$, then
\[
d(\underline{x})=\arg \min_{Y_{i}\in
[-K,K]}|\frac{1}{n}\sum_{i=1}^{n}(\Phi(L_{n})-Y_{i})^{2}|
\]

The stopping rule in RT is decided based on the minimum number of
split in the posterior sample called minsplit. If minsplit $\geq \alpha$
then $\omega_{i}$ will split into two child nodes and if minsplit $<
\alpha$ then $\omega_{i}$ is a leaf node and no more split is
required. Here $\alpha$ is determined by the user, usually it is
taken as 10\% of the training sample size. Now let $\mathcal{T}
=(\widetilde{\Omega}_{1},\widetilde{\Omega}_{2},...)$ be a finite
collection of partitions of a measurement space $\underline{X}$. Let
us define maximal node count of $\mathcal{T}$ as the maximum number
of nodes in any partition $\widetilde{\Omega}$ in $\mathcal{T}$
which can be written as
\[
\lambda(\mathcal{T})=\sup_{\widetilde{\Omega}_{i} \in
\mathcal{T}}|\widetilde{\Omega}_{i}|
\]
Also let, $\Delta (\mathcal{T}, L_{n})=|\{( L_{n})_{\widetilde{\Omega}}:
\widetilde{\Omega} \in \mathcal{T}\}|$ be the number of distinct
partitions of a training sample of size n induced by
partitions in $\mathcal{T}$. Let $\Delta_{n}(\mathcal{T})$ be the growth
function of $\mathcal{T}$ defined as
\[
\Delta_{n}(\mathcal{T})=\sup_{\{L_{n}:|L_{n}|=n\}} \Delta(\mathcal{T},L_{n}).
\]
Growth function of $\mathcal{T}$ is the maximum number of distinct
partitions $(L_{n})_{\widetilde{\Omega}}$ which partition
$\widetilde{\Omega}$ in $\mathcal{T}$ can induce in any training
sample with n observations. For a partition $\widetilde{\Omega}$ of
X, $\widetilde{\Omega} [x \in X]=\{ \omega_{i} \in
\widetilde{\Omega}: x \in \omega \}$ be the node $\omega_{i}$ in
$\widetilde{\Omega}$ which
contains $x$.\\

For consistency of any histogram based regression estimates, the
sub-linear growth of restricted cell counts (see in Eqn. 1),
sub-exponential growth of a combinatorial complexity measure (see in
Eqn. 2) and shrinking cell (see in Eqn. 3) conditions are to be
satisfied \cite{nobel1996histogram,sugumaran2007feature} and these
are as follows :

\begin{equation}
\frac{\lambda(\mathcal{T}_{n})}{n}\rightarrow 0 \quad  \mbox{as}
\quad n \rightarrow \infty.
\end{equation}
\begin{equation}
\frac{log(\triangle_{n}(\mathcal{T}_{n}))}{n}\rightarrow 0 \quad
\mbox{as} \quad n \rightarrow \infty
\end{equation}
and for every $\gamma > 0$ and $\delta \in (0,1)$,
\begin{equation}
\inf_{S\subseteq \mathbb{R}^{p} : \ P(S) \geq 1-\delta} \ P({x:
diam(\tilde{\Omega}_{n}[x]\cap S) > \gamma})\rightarrow 0 \quad
\quad \quad \quad \quad  w.p.\quad 1
\end{equation}

Now we are going to evolve a sufficient regularity condition for the
binary partitioning and regression tree based rule to be universally
consistent. It is to be shown that optimal
regression trees are consistent when the size of the tree grows as
$o(\frac{n}{log(n)})$, where n is the number of training samples and
it implies the regularity conditions of Eqn. (1), (2) and (3).

\begin{theorem}
Suppose $(\underline{X},\underline{Y})$ be a random vector in
$\mathbb{R}^{p} \times [-K,K]$ and $L_{n}$ be the training set of n
outcomes of $(\underline{X},\underline{Y})$. Finally if for every n
and $w_{i}\in \tilde{\Omega}_{n}$, the induced subset
$(L_{n})_{w_{i}}$ contains at least $k_{n}$ of the vectors of
$X_{1},X_{2},...,X_{n}$, then empirically optimal regression trees
strategy employing axis parallel splits are consistent when the size
$k_{n}$ of the tree grows as $o(\frac{n}{log(n)})$.
\end{theorem}

\begin{proof}

Since $\mathcal{T}(k_{n})$ contain all the binary RT partitions
having $k_{n}$ leaves, so $\lambda(\mathcal{T}(k_{n}))=k_{n}$.
Therefore, $\frac{\lambda(\mathcal{T}(k_{n}))}{n}=\frac{k_{n}}{n}$
tends to zero as $n \rightarrow \infty$. Hence,
$\frac{\lambda(\mathcal{T}(k_{n}))}{n} \rightarrow 0$. Thus
condition $(1)$ holds.\\

Now regression trees having $k_{n}$ leaves has $k_{n}-1$ internal
nodes and therefore each partition is based on at most $k_{n}-1$
intersecting half-spaces. And by Cover's theorem
\cite{cover1965geometrical}, any binary split of $\mathbb{R}^{p}$
can divide n points in at most $n^{p}$ ways. So, their intersection
will partition n points in at most $n^{(k_{n}-1)p}$ ways. Thus we
write,
\[
\triangle_{n}(\mathcal{T}(k_{n})) \leq n^{(k_{n}-1)p}
\]
and consequently,
\begin{equation}
\frac{log(\triangle_{n}(\mathcal{T}(k_{n})))}{n} \leq
\frac{p(k_{n}-1)}{n} log(n)
\end{equation}
As, $n \rightarrow \infty$, RHS of equation (4) goes to zero. So
condition $(2)$ holds.

Now $k_{n}=o(\frac{n}{log(n)})$, then for every n and $\omega \in
\tilde{\Omega}_{n}$, the induced subset $(L_{n})_{\omega} \in
\mathcal{T}(k_{n})$ such that for every compact set $V \subseteq
\mathbb{R}^{d}$ we can write $max_{A \in (L_{n})_{\omega_{i}}} diam
(A \cap V) \rightarrow 0$. This implies condition $(3)$ is satisfied and hence the theorem.\\
\end{proof}

\begin{remark}
Thus, it may be noted that regression trees are empirically
consistent when the size $k_{n}$ of the tree grows as
$o(\frac{n}{log(n)})$. It is interesting to know that without
employing sufficient conditions as mentioned in equation (1), (2)
and (3), we can have a universally consistent regression tree
provided the regularity condition of Theorem 1 is satisfied. It is
also noted that with sufficiently large n, the optimal regression
tree will not divide the regions of the feature space on which the
regression function is constant. We further conclude that feature
selection using RT algorithm is justified and RT output will also
play an important role in designing the regression model for
increasing the predictive accuracy of the model. It should also be
mentioned that incorporating RT output as an input feature in ANN,
the dimensionality gets increased, thus the performance of the ANN
model will be improved at a significant rate
\cite{kohonen1988statistical}.
\end{remark}

Our proposed hybrid model has two parts: extracting important
features from the feature space using RT algorithm and building one
hidden layered the ANN model with the important features extracted
using RT along with RT output as another input vector in the ANN
model. In our base model, the dimension of the input layer, denoted
by $d_{m}(\leq{p})$, is the number of important features obtained by
RT + 1. It is also noted that one-hidden layered neural networks
yield strong universal consistency and there is a little theoretical
gain in considering two or more hidden layered neural network
\cite{devroye2013probabilistic}. In hybrid RT-ANN model, we have
used one hidden layer with k neurons. This makes our hybrid model
less complex and less time consuming while running the model. Our
objective is to state the sufficient condition for universal
consistency and then to find out the optimal value of $k$. Before
stating the sufficient conditions for the consistency of the
algorithm and the optimal number of nodes in the hidden layer, let
us define the
following: \\

Let the random variables $\underline{Z}$ and $\underline{Y}$ take
their values from $\mathbb{R}^{d_{m}}$ and $[-K,K]$ respectively.
Denote the measure of $\underline{Z}$ over $\mathbb{R}^{d_{m}}$ by
$\mu$ and $m:\mathbb{R}^{d_{m}}\rightarrow [-K,K]$ be a measurable
function such that $m(Z)$ approximates $Y$. Given the training
sequence $\xi_{n}=\{ (Z_{1},Y_{1}), (Z_{2},Y_{2}),...,(Z_{n},Y_{n})
\}$ of $n$ iid copies of ($\underline{X},\underline{Y}$), the
parameters of the neural network regression function estimators are
chosen such that it minimizes the empirical $L_{2}$ risk =
$\frac{1}{n} \sum_{j=1}^{n}|f(Z_{j})-Y_{j}|^{2}$. We have used
logistic squasher as sigmoid function in neural network which is defined as follows:\\

\textbf{Definition:} A sigmoid function $\sigma(x)$ is called a
logistic squasher if it is non-decreasing, $\lim_{x\rightarrow
-\infty} \sigma(x)= 0$ and $\lim_{x \rightarrow \infty} \sigma(x)=
1$ with $\sigma(x)=\frac{1}{1+exp{(-x)}}$.\\

Now consider the class of neural networks having logistic squasher
with $k$ neurons in the hidden layer with bounded output weights
\[
\mathscr{F}_{n,k}=\Bigg\{
\sum_{i=1}^{k}c_{i}\sigma(a_{i}^{T}z+b_{i})+c_{0} : k \in
\mathbb{N}, a_{i} \in \mathbb{R}^{d_{m}}, b_{i},c_{i} \in
\mathbb{R}, \sum_{i=0}^{k}|c_{i}|\leq \beta_{n} \Bigg\}
\]
and obtain $m_{n} \in \mathscr{F}_{n,k}$ satisfying
\[
\frac{1}{n} \sum_{i=1}^{n}|m_{n}(Z_{i})-Y_{i}|^{2}  \leq \frac{1}{n}
\sum_{i=1}^{n}|f(X_{i})-Y_{i}|^{2}, \mbox{if} \quad f \in
\mathscr{F}_{n,k}
\]
where, $m_{n}$ be a function that minimizes the empirical $L_{2}$
risk in $\mathscr{F}_{n,k}$. It was shown that if $k$ and
$\beta_{n}$
are chosen to satisfy :\\

$k\rightarrow \infty$, $\beta_{n}\rightarrow \infty$,
$\frac{k\beta_{n}^{4}log(k\beta_{n}^{2})}{n} \rightarrow 0$, and
there exists $\delta (> 0)$ such that
$\frac{\beta_{n}^{4}}{n^{1-\delta}} \rightarrow 0$, then

\[
E\int(m_{n}(z)-m(z))^{2}\mu(dz) \rightarrow 0 \quad \mbox{as}\quad n
\rightarrow \infty
\] and consequently,
\[
\int(m_{n}(z)-m(z))^{2}\mu(dz) \rightarrow 0 \quad  \mbox{as} \quad
n \rightarrow \infty
\]
for all distributions of ($\underline{Z},\underline{Y}$) with
$E|Z^{2}| < \infty$, i.e., $m_{n}$ is strongly universally
consistent (For proof see \cite[Theorem
3]{lugosi1995nonparametric}). The rate of convergence of the neural
network with k neurons with bounded output weights has been pointed
out by \cite{gyorfi2006distribution}. This made us conclude that our
proposed model is strongly consistent under the mentioned regularity
condition and our job is to find an optimal choice of k for our
algorithm. To obtain the optimal value of the number of hidden
neuron in the hidden layer of the hybrid RT-ANN model having $d_{m}$
number of input layers, we state the following proposition :
\newline

\begin{prop}
For any fixed $d_{m}$ and training sequence $\xi_{n}$, let
$Y\in[-K,K]$, and $m, f \in \mathscr{F}_{n,k}$, if the neural
network estimate $m_{n}$ satisfies the above-mentioned regularity
conditions of strong universal consistency and $f$ satisfying
$\int_{S_{r}}f^{2}(z)\mu(dz)<\infty$ where, $S_{r}$ is a ball with
radius r centered at 0, then the optimal choice of $k$ is $O \bigg
(\sqrt{\frac{n}{d_{m}log(n)}} \bigg )$.
\end{prop}

\begin{proof}
To prove Proposition 1, we are going to recall the results of the
proofs of universal consistency and rate of convergence \cite[Page
301-326]{gyorfi2006distribution}.

The expression for rate of
convergence of the neural network estimates as defined in
$\mathscr{F}_{n,k}$ is as follows:

\begin{equation}
E\int_{S_{r}}|m_{n}(z)-m(z)|^{2}\mu(dz)= O\bigg(\beta_{n}^{2}\bigg(
\frac{d_{m}log(\beta_{n}{n})}{n}\bigg)^{1/2} \bigg)  \quad
\mbox{since} \quad d_{m} > 1
\end{equation}
Since $\beta_{n} < \mbox{constant} < \infty$ and $d_{m} (>1)$ is
a fixed constant. We can write (5) as
\begin{equation}
E\int_{S_{r}}|m_{n}(z)-m(z)|^{2}\mu(dz)=O\bigg(\sqrt{\frac{d_{m}log(n)}{n}}\bigg)
\end{equation}
It is noted that average estimation error (LHS of (6)) approaches to
$0$ as $n \rightarrow \infty$. \\
To obtain upper bound on the rate of convergence we will have to
impose some more regularity conditions as follows: For every continuous function $f: \mathbb{R}^{d_{m}}\rightarrow
[-K,K]$ satisfying the properties of fourier transformation, as
mentioned in Proposition 1, we can write an expression of the rate
of approximation of neural networks as follows:

\begin{equation} \inf_{f \in
\mathscr{F}_{n,k}}\int_{S_{r}}|f(z)-m(z)|^{2}\mu(dz) =
O\bigg(\frac{1}{k}\bigg)
\end{equation}
It is noted that approximation error (LHS of (7)) approaches to
$0$ as $k \rightarrow \infty$. \\

Since the conditions of strong universal consistency of neural
networks and an additional restriction on $f(z)$ are assumed, we can
now put an upper bound to the estimation error using the
approximation error.

Bringing (6) and (7) together, we have the following:
\[
E\int_{S_{r}}|m_{n}(z)-m(z)|^{2}\mu(dz) - \inf_{f \in
\mathscr{F}_{k,n}}\int_{S_{r}}|f(z)-m(z)|^{2}\mu(dz) = O\bigg(
\sqrt{\frac{d_{m}log(n)}{n}}  + \frac{1}{k}\bigg)
\]
For optimal choice of $k$, the problem reduces to equating
$\sqrt{\frac{d_{m}log(n)}{n}}$ with $\frac{1}{k}$ which gives
$k=O\bigg( \sqrt{\frac{n}{d_{m}log(n)}}\bigg)$
\end{proof}

\begin{remark}
The optimal choice of k is found to be
$O(\sqrt{\frac{n}{d_{m}log(n)}})$ and accordingly we have to choose
the number of hidden neurons in the RT-ANN model. For practical use,
one can use $k=\sqrt{\frac{n}{d_{m}log(n)}}$ for small data sets to
get the desired accuracy of the hybrid RT-ANN model. The application
of this model and its implementation is shown in Section 4. It is
found that the proposed model is highly competitive than the other
conventional regression models in terms of accuracy. The universal
consistency of the proposed hybrid RT-ANN model is justified, and it
also gives the upper bound for the number of hidden neurons in the
model.
\end{remark}

In our proposed model, the selected important features using RT and
RT outputs as input features to the neural network with the number of
hidden neurons to be chosen as $O(\sqrt{\frac{n}{d_{m}log(n)}})$
make a strong universally consistent regression model. It is then
applied to solve Krofta efficiency problem and proved to be highly
effective through experimental analysis as shown in Section 4.

\section{Application} \label{Application}

The current problem of Krofta Efficiency aims to improve the Krofta
fiber-filler recovery by reducing the fiber, filler losses and
improving the paper qualities in terms of opacity, formation and
machine runnability. Krofta is fiber-filler recovery equipment
having floatation cum sedimentation cell. The inlet to Krofta is the
paper machine lean backwater which undergoes froth flotation when
treated with chemical and air
\cite{marques2000use,aldrich2010online}. The chemical forms the
flocks from the fines present in the lean backwater and air helps to
form the cake over Krofta supracell for recovery. Supracell removes
solids through air floatation and sedimentation process. Turbulence
caused by water movement is a significant factor in floatation and
greatly reduces the efficiency of the other types of floatation
units. In conventional, there must always be water movement for the
water to flow from inlet to outlet. With the supracell, the inlet
and outlet are not stationary but are rotating about the center. The
rotation is synchronized so that the water in the tank achieves zero
velocity during floatation \cite{zhou1994role}. This means that the
efficiency of floatation is significantly increased to nearly
maximum theoretical limits. In practical terms, this allows better
clarification in smaller surface areas and a much shallower tank.
Water is processed from the inlet to outlet in 2-3 minutes. Air is
dissolved into the water using Krofta air dissolving tube (ADT)
\cite{krofta1998three} and unclarified water is released through a
valve. The water flows in at the exact center, through a rotary
joint, and into the distribution duct. Coarse air is released
through a vent pipe in the duct. The flow is directed to eliminate
turbulence. Since the inlet distribution is moving forward at the
same speed that the water is flowing out, the water stays in one
spot in the tank without any movement during floatation. The floated
material is recovered from the top surface by means of the Krofta
spiral scoop \cite{krofta1986apparatus}. The scoop is designed to
remove the floated material at the highest possible consistency,
with a minimum surface disturbance. The level of water determines
the consistency of the floated material removed. The Krofta
floatation system removes the solid content in the water by floating
them to the surface for removal. The reason why fiber will float,
even if they are heavier than water, is that small air bubbles
attach themselves to the particles or flocks and make them buoyant
\cite{endo2001vapor}. The process flow diagram of Krofta supracell
is shown in Figure 2, and the mechanism for forming the air bubbles
is as follows:

\begin{itemize}
\item  The water is pressurized to feed to the ADT.

\item  Air is added to the pressurized water and gets dissolved at
this pressure.

\item  The pressure is released after the water passes through a
valve. The water can no longer hold the extra air which was
absorbed, so that small bubbles form spontaneously throughout the
liquid. The bubbles formed are tiny.

\item  A small mastering pulp feeds the chemical flocculants.
Chemicals are used to increase the clarifier efficiency by flocking
out small and colloidal particles, that otherwise would not float or
settle in the clarifier.
\end{itemize}

\begin{figure}[t]
\centering
%\subfloat[]
%\subfig{
%\includegraphics[height=1cm,angle=0,width=.6\linewidth,]{mod_realworld.eps}%
  %\epsfig{file=amazon_sim_final.ps, height=2.5in, width=3in}
%includegraphics{Capture.PNG}
%\includegraphics[width=0.6\textwidth,natwidth=50,natheight=50]{Capture.eps}
\includegraphics[scale=0.65]{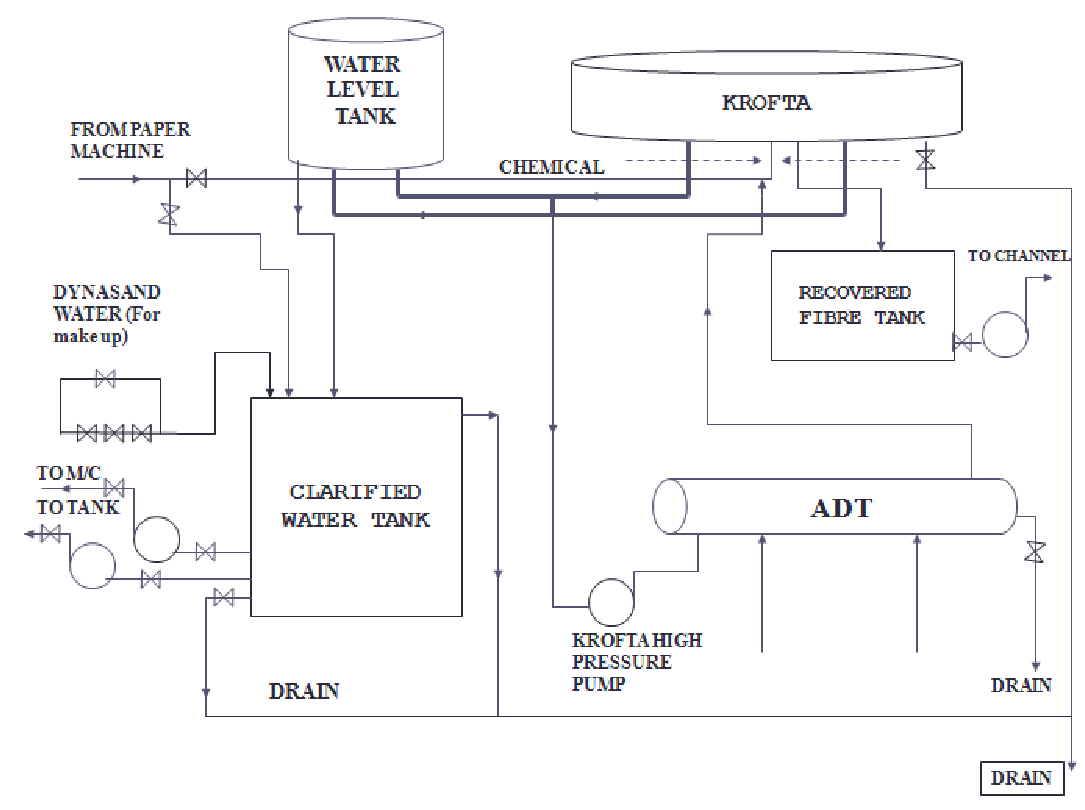}
\caption{Process Flow Diagram of KROFTA Supracell}
\label{figmodRWNGT}
\end{figure}

\subsection{\textbf{Data Collection Plan}}

Initially, the efficiency of Krofta performance which needs
improvement was measured. The efficiency of Krofta
performance\footnote{$ Efficiency = \frac{(Inlet PPM - Outlet
PPM)*100}{Inlet PPM }$} is quantified in terms of Inlet parts per
million (Inlet PPM) and Outlet PPM. To identify the causal
parameters affecting the Krofta efficiency, a failure mode effect
analysis (FMEA) was done with the help of process experts. FMEA is a
useful analytical tool used in the process industry for
understanding possible causes of failures with their impacts and
frequencies. The information collected through FMEA are namely;
potential failure modes, potential failure effects, potential
causes, mechanism of failure and their severity were analyzed. The
potential areas of failures were identified where risk priority
number (RPN) were found to be high. For the potential failure model
where RPN is high, the process experts carried out some minor
experiments to understand the exact reasons for the low efficiency
of Krofta. However, subsequently considering all such potential
areas, a cause and effect diagram was made. The cause and effect
diagram revealing the key relationship among causal variables with
Krofta efficiency is presented in Figure 3. Since not all the
parameters are controllable by the users of the process, with the
help of process experts and further brainstorming, a list of
controllable parameters is prepared. The data were then collected
for two types of tissues, namely tissue 1 and tissue 2, taking
Krofta efficiency percentage as the response and other parameters as
the causal variable.

\begin{figure}[t]
\centering
%\graphicspath{ {images/} }
\includegraphics[scale=0.85]{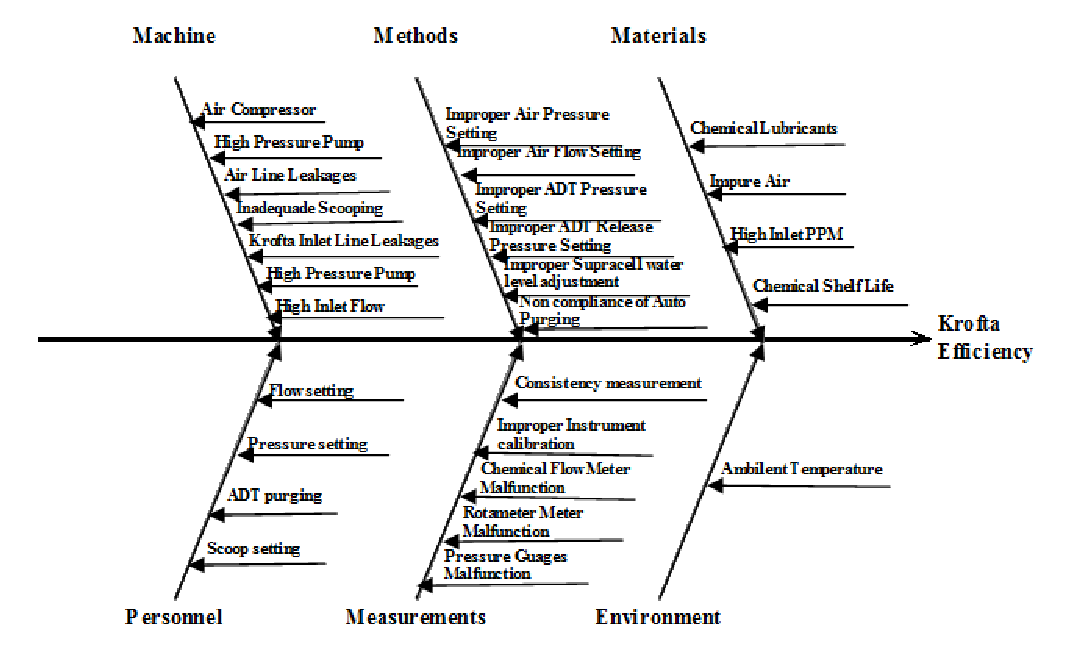}
\caption{Cause and Effect Diagram for Krofta Efficiency}
\end{figure}

\subsection{\textbf{Description of Datasets}}

The dataset collected has the following parameters possibly
affecting the Krofta performance for both the tissues: Inlet Flow,
Water Pressure (water inlet pressure to ADT), Air Pressure, Pressure
of Air-Left, Pressure of Air-Right, Pressure of ADT-D Left, Pressure
of ADT-D Right and Amount of chemical lubricants. A sample dataset
for both the tissues are given in Tables 1 and 2.

\begin{table}[H]
\tiny \centering \caption{Sample Dataset for Tissue 1}
    \begin{tabular}{|c|c|c|c|c|c|c|c|c|c|}
        \hline
        Sl. No. & Inlet Flow & Water Pressure & Air Pressure  & Air-Left  & Air-Right  &  ADT-D  &  ADT-D   &  Amount of &  Recovery  \\
                &            &                &               &           &            &  Left   &  Right   &  chemical  & Percentage \\ \hline
        1       &   2804     &      5.8       &     5.0       &    2.6    &    2.4     &   2.3   &   2.7    &     2.0    &   49.74    \\
        2       &   3200     &      6.0       &     7.0       &    3.4    &    4.0     &   2.2   &   2.7    &     2.5    &   87.13    \\
        3       &   2600     &      6.0       &     5.0       &    2.4    &    1.2     &   3.1   &   4.2    &     1.0    &   96.99    \\
        4       &   3002     &      6.2       &     6.4       &    2.1    &    1.5     &   3.0   &   4.0    &     2.0    &   97.46    \\
        5       &   2899     &      6.2       &     5.7       &    2.0    &    1.2     &   3.1   &   4.0    &     2.0    &   97.91    \\
        6       &   2995     &      6.0       &     6.0       &    1.5    &    4.5     &   4.0   &   4.8    &     4.0    &   28.87    \\
        .       &     .      &       .        &      .        &     .     &     .      &    .    &    .     &      .     &     .      \\
        .       &     .      &       .        &      .        &     .     &     .      &    .    &    .     &      .     &     .      \\
        .       &     .      &       .        &      .        &     .     &     .      &    .    &    .     &      .     &     .      \\
        .       &     .      &       .        &      .        &     .     &     .      &    .    &    .     &      .     &     .      \\
        \hline
    \end{tabular}

\end{table}

\begin{table}[H]
\tiny \centering \caption{Sample Dataset for Tissue 2}
    \begin{tabular}{|c|c|c|c|c|c|c|c|c|c|}
        \hline
        Sl. No. & Inlet Flow & Water Pressure & Air Pressure  & Air-Left  & Air-Right  &  ADT-D  &  ADT-D   &  Amount of &  Recovery  \\
                &            &                &               &           &            &  Left   &  Right   &  chemical  & Percentage \\ \hline
        1       &   1794     &      5.2       &     5.6       &    2.4    &    1.6     &   3.6   &   4.0    &     3.0    &   97.47    \\
        2       &   1703     &      6.2       &     6.0       &    2.9    &    1.0     &   3.0   &   4.2    &     2.0    &   26.94    \\
        3       &   1139     &      6.5       &     6.0       &    1.2    &    1.7     &   3.0   &   4.6    &     2.0    &   33.05    \\
        4       &   1448     &      6.4       &     5.8       &    1.0    &    2.1     &   3.2   &   4.0    &     2.0    &   96.80    \\
        5       &   1614     &      5.5       &     5.0       &    2.0    &    2.1     &   3.8   &   4.7    &     2.0    &   97.01    \\
        6       &   1472     &      6.6       &     6.8       &    3.7    &    3.1     &   5.2   &   4.8    &     4.0    &   97.77    \\
        .       &     .      &       .        &      .        &     .     &     .      &    .    &    .     &      .     &     .      \\
        .       &     .      &       .        &      .        &     .     &     .      &    .    &    .     &      .     &     .      \\
        .       &     .      &       .        &      .        &     .     &     .      &    .    &    .     &      .     &     .      \\
        .       &     .      &       .        &      .        &     .     &     .      &    .    &    .     &      .     &     .      \\
        \hline
    \end{tabular}

\end{table}

\subsection{\textbf{Performance Evaluation}}  \label{Performance Evaluation}

To date, a large number of performance metrics have been proposed
and employed to evaluate the accuracy of the regression model, but
no single performance metric has been recognized as the universal
standard. As a result of this, we need to assess the performance
based on multiple metrics, and it will be interesting to see if
different metrics will give the same ranking for the different
regression models to be tested. The metrics used in this study are:
mean absolute error (MAE), root mean square error (RMSE), mean
absolute percentage error (MAPE), the coefficient of multiple
determination ($R^{2}$) and adjusted $R^{2}$ (Adj($R^{2}$)). The use
of these measures represents different angles to evaluate regression
models. The first two are absolute performance measures while the
third one is a relative measure and the last two are measures of
``goodness of fit" for the regression models. The formulae used for
the performance metrics are as follows:
\[
MAE=\frac{1}{n}\sum_{i=1}^{n}\big |y_{i}-\widehat{y}_{i} \big | ;
RMSE = \sqrt{\frac{1}{n}\sum_{i=1}^{n}(y_{i}-\widehat{y}_{i})^{2}} ;
MAPE=\frac{1}{n}\sum_{i=1}^{n}\bigg |
\frac{y_{i}-\widehat{y}_{i}}{y_{i}} \bigg | ;\] \\
\[R^{2}=1- \bigg
[\frac{\sum_{i=1}^{n}(y_{i}-\widehat{y}_{i})^{2}}{\sum_{i=1}^{n}(y_{i}-\overline{y})^{2}}
\bigg ] ; Adj(R^{2})=1-\bigg [ \frac{(1-R^{2})(n-1)}{n-k-1} \bigg ]
\]

where, $y_{i},\overline{y},\widehat{y}_{i}$ denote the actual value,
average value and predicted value of the dependent variable,
respectively for the $i^{th}$ instant. Here $n, k$ denote the number
of data used for performance evaluation and the number of
independent variables, respectively. The lower the value of MAE,
RMSE, and MAPE and the higher the value of $R^{2}$ and adjusted
($R^{2}$), the better the model is.

\subsection{\textbf{Analysis of Results}}

We have shuffled the observations of the Krofta efficiency dataset
randomly and split it into training and testing data sets in a ratio
of 70 : 30. Each experiment is repeated 5 times with different
randomly assigned training and test sets and we will finally report
the averages of the performance metrics observed over 5 times
validations in Table 3.

With the data sets described in Section 4.2, a few popular
regression  models such as multiple linear regression (MLR),
Stepwise Regression, partial least square (PLS) Regression, SVR,
multiple adaptive regression spline (MARS) and neural tree model
were first applied but resulted in very low $R^{2}$ and $Adj(R^{2})$
and very high MAE, RMSE and MAPE. Then we started exploring
conventional nonparametric models such as RT and ANN which perform
better than the parametric models. Then we applied our proposed
hybrid RT-ANN model. All the methods were implemented in the R
Statistical package on a PC with 2.1 GHz processor and 8 GB memory.
We first select important features using RT algorithm which
indicates that Krofta efficiency is significantly dependent on Water
Pressure, Air Pressure, Inlet Flow and ADT-D left for Tissue 1. And
for Tissue 2, the important features are Water Pressure, Air
Pressure, ADT-D Left and ADT-D Right. Then we run our base model
with the selected important features along with RT output as shown
in Figures 4 and 5. The number of hidden layers was chosen to be
$\sqrt{\frac{n}{d_{m}log(n)}}$ where $d_{m}=5$ for this case.
Performance metrics as defined in Section 4.3 are computed for each
of the competitive models. Table 3 shows that the proposed hybrid
RT-ANN model outperforms the others in terms of the
performance metrics.\\

Based on the model we further created an experimental design to
obtain the optimal level of the tuning parameters. Final
recommendations based on the results of the design of experiments
were implemented in the process to monitor the Krofta efficiency.
However, we have discussed here only the proposed model and its
accuracy level compared to other relevant state-of-the-art models.
Our model helped the manufacturing process industry to achieve an
efficiency level of about 85\% from the existing level of about 60\%
to improve the Krofta supracell recovery percentage.

\begin{table}[H]
%\centering
\caption{Quantitative measures of performance for different
regression models.}
    \begin{tabular}{|c|c|c|c|c|c|c|}
        \hline
        Regression Models                             & Tissue Type     & MAE     & RMSE     & MAPE     & $R^{2}$    & Adj($R^{2}$) \\ \hline
        \multirow{2}{*}{Multiple Linear Regression}   & Tissue 1        & 17.92   & 23.32    & 48.40    & 26.95      & 17.30        \\
                                                      & Tissue 2        & 11.67   & 16.94    & 25.21    & 56.70      & 37.83        \\ \hline
        \multirow{2}{*}{Stepwise Regression}          & Tissue 1        & 18.31   & 23.51    & 48.57    & 27.00      & 09.10        \\
                                                      & Tissue 2        & 12.04   & 17.09    & 26.01    & 55.97      & 42.65        \\ \hline
        \multirow{2}{*}{PLS Regression}               & Tissue 1        & 17.50   & 22.40    & 42.59    & 26.95      & 17.49        \\
                                                      & Tissue 2        & 14.47   & 20.16    & 30.19    & 25.56      & 20.19        \\ \hline
        \multirow{2}{*}{MARS}                         & Tissue 1        & 16.10   & 20.26    & 38.86    & 45.39      & 29.54        \\
                                                      & Tissue 2        & 15.71   & 19.86    & 34.16    & 40.49      & 30.57        \\ \hline
        \multirow{2}{*}{RT}                           & Tissue 1        & 06.31   & 10.54    & 16.35    & 85.07      & 81.21        \\
                                                      & Tissue 2        & 07.06   & 10.29    & 15.03    & 84.73      & 81.56        \\ \hline
        \multirow{2}{*}{ANN}                          & Tissue 1        & 05.15   & 08.55    & 08.88    & 90.18      & 87.65        \\
                                                      & Tissue 2        & 08.27   & 11.54    & 17.45    & 80.25      & 76.05        \\ \hline
        \multirow{2}{*}{SVR}                          & Tissue 1        & 05.47   & 08.97    & 08.89    & 86.78      & 82.13        \\
                                                      & Tissue 2        & 08.46   & 12.17    & 17.31    & 81.23      & 78.72        \\ \hline
        \multirow{2}{*}{Neural tree}                  & Tissue 1        & 04.95   & 06.44    & 06.12    & 92.36      & 90.53        \\
                                                      & Tissue 2        & 06.78   & 09.89    & 15.93    & 85.79      & 83.09        \\ \hline
        \multirow{2}{*}{Hybrid RT-ANN}                & Tissue 1        & 03.45   & 04.89    & 06.87    & 96.79      & 95.32        \\
                                                      & Tissue 2        & 05.91   & 08.60    & 12.67    & 88.84      & 87.75        \\ \hline
    \end{tabular}
\end{table}

\begin{figure}[H]
\centering
%\graphicspath{ {images/} }
\includegraphics[scale=0.25]{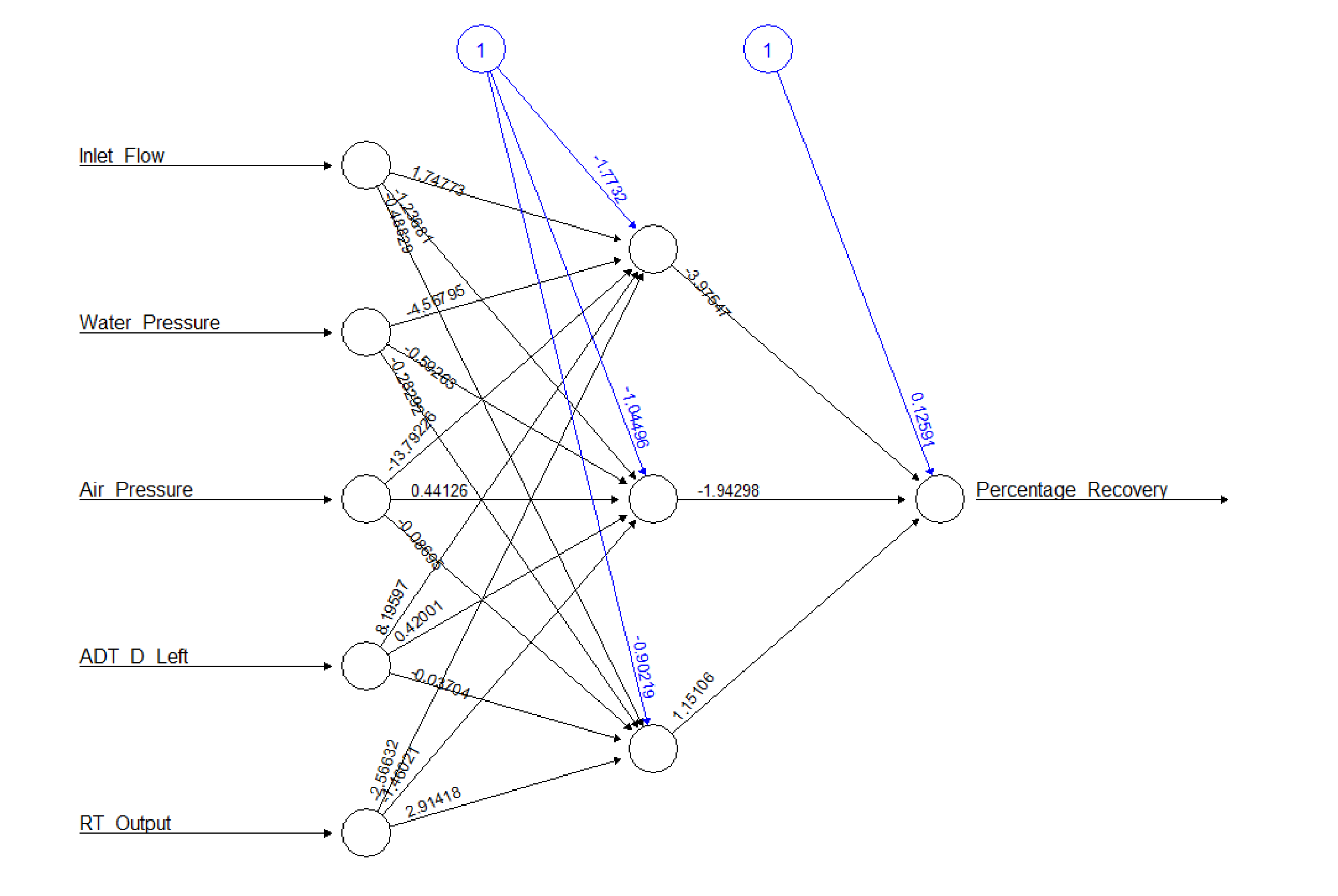}
\caption{Hybrid RT-ANN model for Tissue 1}
\end{figure}

\begin{figure}[H]
\centering
%\graphicspath{ {images/} }
\includegraphics[scale=0.20]{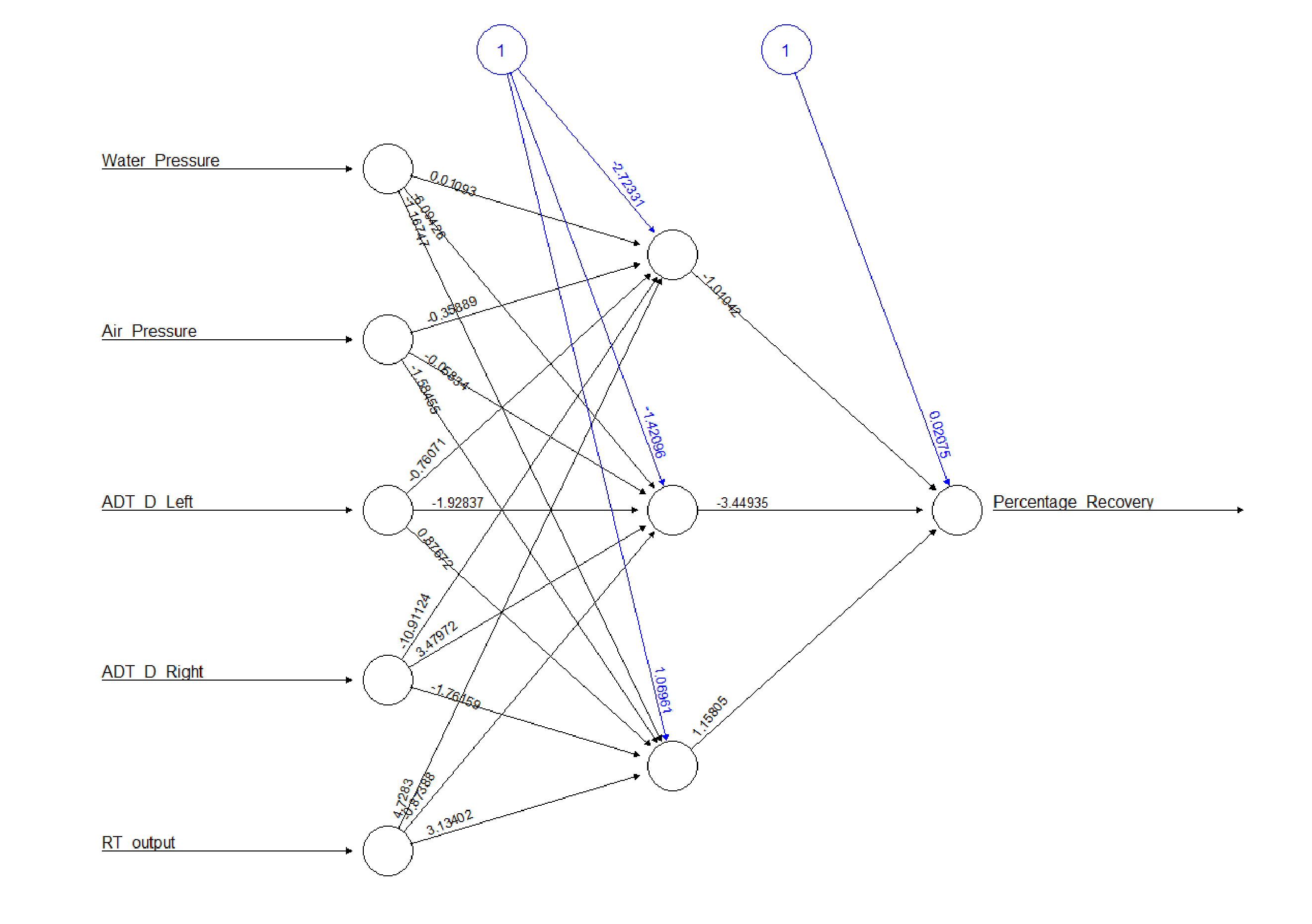}
\caption{Hybrid RT-ANN model for Tissue 2}
\end{figure}

\section{Conclusions and Discussions}

The purpose of this article is to develop a model for selection of
critical causal variable(s) of a manufacturing process. Our study
presented a hybrid RT-ANN model that combines both the neural
network and RT which gives more accuracy than all other traditional
models as shown in Table 3 to address the complicated problem of
Krofta efficiency improvement. We have found RT to be the optimal
technique for the selection of tuning parameters and found hybrid
RT-ANN model to be the optimal regression model for accurately
predicting the Krofta recovery percentage. Consequently, the hybrid
RT-ANN successfully demonstrated the best performance and offered a
practical solution to the problem of finding optimal levels for
selecting the tuning parameters to improve the Krofta efficiency.
Though the hybrid model developed was mainly to improve the
efficiency of Krofta, its proven theoretical consistency makes it
robust and hence can generally be applied to other similar
situations. We have also provided an upper bound for the number of
hidden neurons during the application of ANN. The proven model, when
applied to a complex chemical process for efficiency improvement,
provided the required causal variables. These variables were then
optimized using the design of experiments to achieve an efficiency
level of about 85\% from the existing efficiency level of about
60\%. Since the accuracy of the proposed model is quite high and it
is consistent, the model may be applied to other complex problems as
well. In the application of the Krofta efficiency problem, not only
the efficiency is enhanced, but it also resulted in lesser wastage
and provided a significant economic gain for the organization. It
also helped the organization to be more environmentally friendly as
well.

%\section*{Acknowledgements}

\bibliographystyle{elsarticle-num}

\bibliography{asm}

\end{document}